%% file: main.tex
\documentclass[11pt,english]{article}

\input{auxil/preamble.tex}

\title{ Generalized Automatic Least Squares:\\ Efficiency Gains from Misspecified Heteroscedasticity  Models}
\author{
\setcounter{footnote}{1}
 Bulat Gafarov\thanks{
	Department of Agricultural and Resource Economics, University of California, Davis.
Email: \href{bgafarov@ucdavis.edu}{bgafarov@ucdavis.edu}. I would like to thank Colin Cameron, Takuya Ura and  Kaspar W\"uthrich for helpful comments.} 
}

\begin{document}

\maketitle

\input{sections/abstract}

\input{sections/intro}

\input{sections/model}

\bibliographystyle{ecta}
\bibliography{main}
 


\end{document}

%% file: auxil/preamble.tex
\usepackage[T1]{fontenc}
\usepackage[margin=1in]{geometry}
\usepackage[latin9]{inputenc}
\usepackage{color}
\usepackage[dvipsnames]{xcolor}
\usepackage{enumerate}
\usepackage{enumitem} 
\usepackage{hyperref}
\definecolor{darkblue}{rgb}{0, 0, 0.5}
\hypersetup{
     colorlinks = true,
     citecolor = Maroon,
     urlcolor = darkblue,
     linkcolor = darkblue
}

\usepackage{amsmath}
\usepackage{amsthm}
\usepackage{amssymb}
\usepackage{cases}
\usepackage{bm}
\usepackage[authoryear]{natbib}
\usepackage{setspace}
\usepackage{graphicx}
\usepackage{float}
\usepackage[flushleft]{threeparttable}
\usepackage{pdflscape}
\usepackage{multicol}

\usepackage{auxil/macros}

\makeatletter
\theoremstyle{plain}

  \theoremstyle{plain}

\theoremstyle{plain}
\newtheorem{thm}{\protect\theoremname}
\newtheorem*{thm*}{\protect\theoremname}

  \theoremstyle{plain}

  \theoremstyle{remark}

\makeatother

\usepackage{babel}
  \providecommand{\assumptionname}{Assumption}
  \providecommand{\claimname}{Claim}
  \providecommand{\lemmaname}{Lemma}
  \providecommand{\remarkname}{Remark}
\providecommand{\theoremname}{Theorem}
  \providecommand{\examplename}{Example}
\providecommand{\corname}{Corollary}

%% file: sections/abstract.tex
\begin{abstract}

    \linespread{1.2}

It is well known that in the presence of heteroscedasticity ordinary least squares estimator is not efficient. 
I propose a generalized automatic least squares estimator (GALS) that makes partial correction of heteroscedasticity based on a (potentially) misspecified model without a pretest. 
Such an estimator is guaranteed to be at least as efficient as either OLS or WLS but can provide some asymptotic efficiency gains over OLS  if the misspecified model is approximately correct. 
If the heteroscedasticity model is correct, the proposed estimator achieves full asymptotic efficiency.
The idea is to frame moment conditions corresponding to OLS and WLS squares based on miss-specified heteroscedasticity  as a joint generalized method of moments estimation problem.
The resulting optimal GMM estimator is equivalent to a feasible GLS with estimated weight matrix.
I also propose  an optimal GMM variance-covariance estimator for GALS to account  for any remaining heteroscedasticity in the residuals.

\medskip


\medskip


\end{abstract}

%% file: sections/intro.tex
\section{Introduction}

It is well known that in the presence of heteroscedasticity ordinary least squares estimator is not efficient (see Section 8.4 in  \cite{wooldridge2019introductory}; early proposals were made in \cite{cragg1983more},). 
I propose a generalized automatic least squares estimator (GALS) that makes partial correction of heteroscedasticity based on a (potentially) misspecified model. 
Such  estimator is guaranteed to be at least as asymptotically efficient as either OLS or WLS but can provide some efficiency gains if the misspecified model is approximately correct. 
If the heteroscedasticity model is correct, the proposed estimator achieves full asymptotic efficiency.
The idea is to frame moment conditions corresponding to OLS and WLS squares based on miss-specified heteroscedasticity  as a joint generalized method of moments estimation problem.
The resulting optimal GMM estimator is equivalent to a GLS with estimated weight matrix.
I also propose heteroscedasticity-robust GMM variance-covariance formula to account  for   any remaining heteroscedasticity in the residuals.

The procedure works in three steps. First, one need to  estimate the usual OLS residuals. Then one need to estimate the approximate model for heteroscedasticity. Third, one computes GALS estimator and its covariance matrix using closed-form formulas. All three steps do not involve any optimization and are very easy to compute. 

In a closely related recent paper  \cite{romano2017resurrecting} provide a proof that WLS with estimated (potentially misspecified) heteroscedasticity is heteroscedasticity robust if one applies  \cite{white1980heteroskedasticity} standard errors for the inference. 
Further, they suggest to make a pretest to determine whether  use  OLS or WLS. 
\cite{romano2017resurrecting}  acknowledge that for some misspecified models of heteroscedasticity, WLS is less efficient than OLS, in which case their  resulting adaptive least squares (ALS) procedure would be less efficient than OLS.
My proposal instead is to avoid the pretest and use GALS estimator, which is at least as efficient as OLS or WLS and can even achieve higher efficiency that both of the two estimators.

Other researchers also considered efficiency gains of GMM over OLS  from additional restrictions on higher order moments \citep{im2008more}. In contrast, I do not assume additional properties of the regression residuals beyond the standard exogeneity condition and existence of approximate non-trivial heteroscedasticity model. ( The true model can be homoscedastic in which case GALS is asymptotically equivalent to OLS; the WLS moments will be will be  ignored   asymptotically. )

%% file: sections/model.tex
\section{Model and the procedure}

Consider a linear regression model 
\begin{equation}
    Y_i = X_i'\beta + \varepsilon_i. 
\end{equation}
As usual, $E(\varepsilon_i|X_i)=0$, the vector of regressors $X_i$ has dimension $p$. 
Suppose that $E(\varepsilon_i^2|X_i)=\sigma_\star^2(X_i)= \exp(\sum_{k=1}^\infty \delta_k p_k(X_i))$, where $p_k(\cdot)$ are approximating basis functions like Chebyshev polynomials. 
As long as the heteroscedasticity function is square integrable, we can choose, for example, trigonometric polynomial or other orthogonal basis functions on the support of $X_i$.

Suppose that we pick a flexible model for heteroscedasticity given by formula  $\sigma^2(X_i)$ (for example, $\sigma^2(X_i)= \exp(\sum_{k=1}^K \delta_k p_k(X_i))$).
It is well known that WLS procedure corresponds to moment conditions
\begin{equation}
    E(\varepsilon_i \frac{X_i}{\sigma^2(X_i)}) = 0.\label{eq:WLS}
\end{equation}
If  the functional form is correct, i.e. $\sigma^2(X_i) =  =\sigma_\star^2(X_i)$, then WLS based on \eqref{eq:WLS} would be asymptotically  efficient for estimation of $\beta$ by the Gauss-Markov theorem.
Of course, OLS estimator also corresponds to GMM estimator with moment conditions
\begin{equation}
    E(\varepsilon_i X_i) = 0.\label{eq:OLS}
\end{equation}
Depending on how well we model heteroscedasticity, choice of \eqref{eq:WLS} may or may not improve efficiency over \eqref{eq:OLS}. 
However, if we consider joint GMM estimation with optimal GMM-weights  based on both \eqref{eq:WLS} and \eqref{eq:OLS}, we are guaranteed to be at least as efficient as an estimator that uses only moments \eqref{eq:OLS}(OLS estimator) or moments \eqref{eq:WLS} ( inefficient WLS estimator).
So if our model of heteroscedasticity $\sigma^2(X_i)$ was approximately correct, the joint GMM will be more efficient than OLS by the properties of the optimal GMM estimators.

\subsection{Description of the estimator}

Consider a joint GMM estimator based on both \eqref{eq:WLS} and \eqref{eq:OLS}.
It turns out that it is itself a GLS estimator, which automatically adapts to heteroscedasticity of the residuals (hence the name choice).   It is given by an optimization problem 
\begin{align}
    \hat\beta^{GALS} = \argmin_b \frac{1}{n}\begin{pmatrix}
   \varepsilon'  X^\prime  \\
   \varepsilon'  D  X ^\prime 
\end{pmatrix}  W  \begin{pmatrix}
    X \varepsilon\\
    X D \varepsilon
\end{pmatrix} , 
\end{align}
where $ \varepsilon$ is the vector of residuals, $X$  is  matrix of regressors  with dimension $(k\times n)$, $W$ is the optimal weight matrix  (it is given explicitly in the next subsection), $D = diag(\frac{1}{\sigma^2(X_1)},\dots,\frac{1}{\sigma^2(X_n)})$.
Indeed, after simplification the GLS form becomes explicit, 
\begin{align}
   \hat\beta^{GALS}  &= \argmin_b  \frac{1}{n} \varepsilon' (X' W_{11} X +   X' W_{12} X D +  D X' W_{21} X   + D X' W_{22} X D)  \varepsilon \\
    &= \argmin_b    (  Y - X'b )' A_n  (  Y - X'b ) =  (X   A_n X')^{-1} (X   A_n Y) 
\end{align}
where $A_n = \frac{1}{n}  (X'W_{11} X +   X' W_{12} X D +  D X' W_{21} X   + D X' W_{22} X D)$ is a symmetric $n\times n$ matrix .

\begin{thm} 
Suppose that data is i.i.d., matrix $G = ( E X_iX'_i ,  E \frac{1}{\sigma^2(X_i)} X_i X'_i )'$ has full rank, moment conditions \eqref{eq:WLS}-\eqref{eq:OLS} have a non-degenrate covariance matrix $\Omega$. Then  $ \sqrt{n}(\hat\beta^{GALS}  - \beta ) \overset{d}{\to}N(0, ( G'   W G)^{-1})$.
Moreover, $  \hat\beta^{GALS} $ is weakly more efficient than $ \hat\beta^{OLS}$ and $ \hat\beta^{WLS}$ (i.e. differences between the corresponding asymptotic covariance matrices are positive definite).

\end{thm}

\begin{proof} 
  Follows from the efficiency of GMM with optimal weights $W$  (see Theorem 3.4 in \cite{hall2004generalized}).
\end{proof}
In practice, of course, $A_n$ has to be replaced with a consistent estimator, which is discussed in the next subsection. 
\subsection{Computation of matrix $A_n$}
As usual, the optimal GMM weight matrix $W$ is equal to inverse of the covariance matrix  $\Omega$ of the moment conditions, which is given by
\begin{align}
        \Omega = \begin{pmatrix}
  E \varepsilon^2_i X_iX'_i&  E \frac{1}{\sigma^2(X_i)}\varepsilon^2_i X_i X'_i\\
  E \frac{1}{\sigma^2(X_i)}\varepsilon^2_i X_i X'_i &  E \frac{1}{\sigma^4(X_i)}\varepsilon^2_i X_i X'_i
\end{pmatrix}.
\end{align}
Suppose that matrices $  \Omega_{11} = E \varepsilon^2_i X_iX'_i$ and $(\Omega_{22}- \Omega_{12}  \Omega_{11}^{-1}  \Omega_{12}) $ are invertible. 
Then using Theorem 2.1 from \cite{lu2002inverses} we can obtain eplicit block representaiton of the inverse,
{\tiny
\begin{align}
       W = \Omega^{-1} =  \begin{pmatrix}
 \Omega_{11}^{-1} + \Omega_{11}^{-1} \Omega_{12}(\Omega_{22}- \Omega_{12}  \Omega_{11}^{-1}  \Omega_{12})^{-1} \Omega_{12}\Omega_{11}^{-1}  
 & - \Omega_{11}^{-1} \Omega_{12}(\Omega_{22}- \Omega_{12}  \Omega_{11}^{-1}  \Omega_{12})^{-1}  \\
 -(\Omega_{22}- \Omega_{12}  \Omega_{11}^{-1}  \Omega_{12})^{-1} \Omega_{12} \Omega_{11}^{-1}  &  (\Omega_{22}- \Omega_{12}  \Omega_{11}^{-1}  \Omega_{12})^{-1} 
\end{pmatrix}
\end{align}
}

So matrix $A_n$ takes explicit form
\begin{align}
    A_n =& \frac{1}{n} \big(X' (\Omega_{11}^{-1} + \Omega_{11}^{-1} \Omega_{12}(\Omega_{22}- \Omega_{12}  \Omega_{11}^{-1}  \Omega_{12})^{-1} \Omega_{12}\Omega_{11}^{-1} ) X   \\
    &- X' ( \Omega_{11}^{-1} \Omega_{12}(\Omega_{22}- \Omega_{12}  \Omega_{11}^{-1}  \Omega_{12})^{-1}) X D   \\
    &-  D X'( \Omega_{11}^{-1} \Omega_{12}(\Omega_{22}- \Omega_{12}  \Omega_{11}^{-1}  \Omega_{12})^{-1}) X   \\
    &+ D X'(\Omega_{22}- \Omega_{12}  \Omega_{11}^{-1}  \Omega_{12})^{-1} X D\big)
\end{align}

A feasible version of this matrix is using blocks from 

\begin{align}
        \hat \Omega =\frac{1}{n} \begin{pmatrix}
  \sum_{i=1}^{n}\hat\varepsilon^2_i X_iX'_i&  \sum_{i=1}^{n} \frac{1}{\hat\sigma^2(X_i)}\hat\varepsilon^2_i X_i X'_i\\
  \sum_{i=1}^{n} \frac{1}{\hat\sigma^2(X_i)}\hat\varepsilon^2_i X_i X'_i &  \sum_{i=1}^{n} \frac{1}{\hat\sigma^4(X_i)}\hat\varepsilon^2_i X_i X'_i
\end{pmatrix}
\end{align}
instead of $\Omega$ and  $\hat D  = diag(\frac{1}{\hat\sigma^2(X_1)},\dots,\frac{1}{\hat\sigma^2(X_n)})$ instead of $D$. 
Here $\hat\sigma^2(X_i) = \exp(\sum_{k=1}^K \hat\delta_k p_k(X_i) )$ with $\hat \delta_k$ given by OLS regression of $\log \hat \varepsilon_i^2$ on $( p_1(X_i),\dots, p_k(X_i))$.
The residuals $  \hat \varepsilon_i $ come from OLS regression of $Y_i$ on $X_i$.
Under regularity conditions on the moments that ensure consistency $\hat\delta_k$ and of $\hat\Omega$, one can replace $A_n$ with $\hat A_n $ in Theorem 1.


%% file: main.bbl
\begin{thebibliography}{7}
\newcommand{\enquote}[1]{``#1''}
\expandafter\ifx\csname natexlab\endcsname\relax\def\natexlab#1{#1}\fi

\bibitem[\protect\citeauthoryear{Cragg}{Cragg}{1983}]{cragg1983more}
\textsc{Cragg, J.~G.} (1983): \enquote{More efficient estimation in the
  presence of heteroscedasticity of unknown form,} \emph{Econometrica: Journal
  of the Econometric Society}, 751--763.

\bibitem[\protect\citeauthoryear{Hall}{Hall}{2004}]{hall2004generalized}
\textsc{Hall, A.~R.} (2004): \emph{Generalized method of moments}, OUP Oxford.

\bibitem[\protect\citeauthoryear{Im and Schmidt}{Im and
  Schmidt}{2008}]{im2008more}
\textsc{Im, K.~S. and P.~Schmidt} (2008): \enquote{More efficient estimation
  under non-normality when higher moments do not depend on the regressors,
  using residual augmented least squares,} \emph{Journal of Econometrics}, 144,
  219--233.

\bibitem[\protect\citeauthoryear{Lu and Shiou}{Lu and
  Shiou}{2002}]{lu2002inverses}
\textsc{Lu, T.-T. and S.-H. Shiou} (2002): \enquote{Inverses of 2$\times$ 2
  block matrices,} \emph{Computers \& Mathematics with Applications}, 43,
  119--129.

\bibitem[\protect\citeauthoryear{Romano and Wolf}{Romano and
  Wolf}{2017}]{romano2017resurrecting}
\textsc{Romano, J.~P. and M.~Wolf} (2017): \enquote{Resurrecting weighted least
  squares,} \emph{Journal of Econometrics}, 197, 1--19.

\bibitem[\protect\citeauthoryear{White}{White}{1980}]{white1980heteroskedasticity}
\textsc{White, H.} (1980): \enquote{A heteroskedasticity-consistent covariance
  matrix estimator and a direct test for heteroskedasticity,}
  \emph{Econometrica: journal of the Econometric Society}, 817--838.

\bibitem[\protect\citeauthoryear{Wooldridge}{Wooldridge}{2019}]{wooldridge2019introductory}
\textsc{Wooldridge, J.~M.} (2019): \emph{Introductory Econometrics: A Modern
  Approach}, Cengage Learning.

\end{thebibliography}
